\documentclass{ifacconf}

\usepackage{graphicx}      
\usepackage{natbib}        

\usepackage{cmap}
\usepackage{xcolor}
\usepackage{amsfonts}
\usepackage{amsmath}
\usepackage{amssymb}

\newtheorem{definition}{Definition}
\newtheorem{lemma}{Lemma}
\newtheorem{remark}{Remark}
\newtheorem{theorem}{Theorem}
\newtheorem{proof}{Proof}

\begin{document}
\begin{frontmatter}

\title{Fast Fixed-time Convergence in Nonlinear Dynamical Systems}


\author[First]{Igor B. Furtat} 

\address[First]{IPME RAS, St. Petersburg, Boljshoy pr. V.O., 61, Russia (e-mail: cainenash@mail.ru).}

\begin{abstract}                
A fast convergence in a fixed-time of solutions of nonlinear dynamical systems, for which special requirements are satisfied on the derivative of a quadratic function calculated along the solutions of the system, is proposed. 
The conditions for the system solutions to converge to zero and to a given region within a fixed-time are obtained.
To achieve fast convergence, a negative power is applied to the derivative of a quadratic function within a specific time interval during the evolution of the system.
The application of the proposed results to the design of control laws for arbitrary order linear plants using the backstepping method is considered.
All the main results are accompanied by numerical modelling and a comparison of the proposed solutions with some existing ones.
\end{abstract}


\end{frontmatter}

\section{Introduction}

The search for control laws that ensure solutions converge in a finite-time in the closed-loop system has attracted the attention of many researchers in the field of automatic control theory.
Such problems are relevant in the fields of aviation, robotics and the chemical industry, where completing a task in less than the critical time is important, see, for example, \cite{Bhat98,Ferrara98,Bartolini03,Chen25}.

The issue of finite-time stability was first discussed in the seminal publications \cite{Erugin51, Zubov64, Roxin65, Utkin92}, where, in particular, it was proposed to use piecewise constant control laws based on a sign function.
When these control laws are used in practice, high-frequency oscillations with large amplitude can arise in the control signal due to chattering. 
This is associated with the presence of noises in the measurement channel and delay, the use of approximate calculations and approximate sign function, etc.

A great deal of work has currently been done in developing control laws that stabilise a system in a finite-time. 
Algorithms based on twisting and supertwisting, as well as those based on the homogeneity properties of differential equations, have been proposed, for example, in \cite{Haimo86,Emelyanov86,Bhat00,Davila05,Moulay06,Orlov09,Bejarano10,Orlov11,Efimov11,Utkin13}, and these are just a few of them.
Here, continuous control laws are employed to stabilise the system and eliminate large-magnitude oscillations in control signal under chattering conditions.

If we denote by $T$ the time after which all solutions of a dynamic system are equal to zero and $x_0$ is the system initial condition, then in \cite{Haimo86,Emelyanov86,Utkin92,Bhat00,Davila05,Moulay06,Orlov09,Bejarano10,Orlov11,Efimov11,Utkin13} the condition of convergence in a finite-time is satisfied in the form $T(x_0) < \infty$.

In contrast to \cite{Haimo86,Emelyanov86,Utkin92,Bhat00,Davila05,Moulay06,Orlov09,Bejarano10,Orlov11,Efimov11,Utkin13}, in \cite{Polyakov12}, the proposed solutions guarantee the convergence to the equilibrium position in a fixed-time, i.e. there exists $T_{\max}>0$ such that for any $x_0$ we have $T(x_0) \leq T_{\max}$. 
Similar effects have been previously described in \cite{James90,Efimov06,Raff08,Cruz10}.

Currently, the fixed-time control method has become widely used. 
A comprehensive survey reviewing recent fixed-time stability theories and their applications in multi-agent systems, including consensus, formation, and containment control is presented in \cite{Zuo2022}. 
In \cite{Wang2022}, a fixed-time sliding mode control scheme for a flexible spacecraft, explicitly addressing vibration suppression and external disturbances is proposed.
In \cite{Chen2023}, a distributed fixed-time optimization algorithm for multi-agent systems that explicitly accounts for practical input saturation constraints. 
In \cite{Liu2023} a fixed-time control with fuzzy logic systems and barrier Lyapunov functions to handle unknown nonlinearities and ensure states never violate predefined constraints. 
In \cite{Li2024}, a co-design framework that combines fixed-time consensus control with a coding mechanism to secure multi-agent systems against deception attacks on communication channels is considered. 
In \cite{Basin2024}, a novel fixed-time integral sliding mode controller is designed for linear systems that eliminates the reaching phase and ensures robustness from the initial time instant. 
All of the above methods are based on the results of \cite{Polyakov12}, whereby fixed-time convergence is defined as in \cite{Polyakov12}. 
This paper will propose new conditions that will enable us to reduce the convergence time.

In this paper, novel conditions are proposed that guarantee fast convergence in a fixed-time, which allows to reduce the value of $T_{\max}$ compared to \cite{Polyakov12}.
Three lemmas with the fast convergence in a fixed-time to a given set and to zero are considered.
Then, the proposed results are applied to the design of the control law using the backstepping method, see \cite{Khalil00}.
All the obtained results are accompanied by numerical modelling and comparison of the proposed solutions with some existing ones.

The paper is organised as follows. 
Section \ref{Sec_Main_Result} considers preliminary definitions and lemma as well as novel three lemmas about fast convergence in a fixed-time to a given set and to zero. 
Also in this section, an example is presented numerical comparisons of the proposed result with some existing ones with linear control law, finite-time control laws, fixed-time control law, and hyperexponential control law.
Section \ref{Sec_Control} considers the application of the results of Section \ref{Sec_Main_Result} to design of a control law for arbitrary order linear systems. 
The example in Section \ref{Sec_Control} contains a comparison of the proposed control law with a linear one under different initial conditions and disturbances.
Section \ref{Concl} contains the main conclusions of the paper.


\section{Preliminary Information and Some Further Extensions}
\label{Sec_Main_Result}

Consider the system
\begin{equation}
\label{eq_main}
\begin{array}{lll}
\dot{x}=f(t,x),~~~x(0)=x_0,
\end{array}
\end{equation}
where $t \geq 0$,
$x \in \mathbb R^n$ and $f: \mathbb [0; +\infty) \times \mathbb R^n \to \mathbb R^n$ is a nonlinear function that can be discontinuous.
The solutions of \eqref{eq_main} are understood in the sense of Filippov (\cite{Filippov85}).
Assume the origin is an equilibrium point of \eqref{eq_main}.


\begin{definition}[\cite{Bhat00,Orlov09}]
The origin of \eqref{eq_main} is said to be globally finite-time stable if it is globally asymptotically stable and any solution $x(t,x_0)$ of \eqref{eq_main} reaches the equilibria at some finite-time, i.e., $x(t,x_0)=0$, $\forall t \geq T(x_0)$, where $T: \mathbb R^n \to [0; +\infty)$ is the settling-time function.
\end{definition}

\begin{definition}[\cite{Polyakov12}]
The origin of \eqref{eq_main} is said to be fixed-time stable if it is globally finite-time stable and settling-time function $T(x_0)$ is bounded, i.e. $\exists T_{\max}>0: T(x_0) \leq T_{\max}, \forall x_0 \in \mathbb R^n$.
\end{definition}

\begin{definition}[\cite{Orlov09}]
The set $\mathcal M$ is said to be globally finite-time attractive for \eqref{eq_main} if for any solution $x(t,x_0)$ of \eqref{eq_main} reaches $\mathcal M$ at some finite-time $t=T(x_0)$ and remains there $\forall t \geq T(x_0)$.
\end{definition}

\begin{definition}[\cite{Polyakov12}]
The set $\mathcal M$ is said to be fixed-time attractive for \eqref{eq_main} if it is globally finite-time attractive and settling-time function $T(x_0)$ is globally bounded by some number $T_{\max}>0$.
\end{definition}

\begin{lemma}[\cite{Polyakov12}]
\label{lem00}
If there exists a continuous radially unbounded positive definite function $V: \mathbb R^n \to [0; +\infty)$ such that
1) $V(x)=0 \Rightarrow x \in \mathcal M$ ; 2) any solution $x(t)$ of \eqref{eq_main} satisfies the inequaity 
\begin{equation}
\label{eq_deriv_Pol_01}
\begin{array}{lll}
\mathcal D^{+}  V(x(t)) \leq - \left[\alpha V^{p}(x(t)) + \chi V^{l}(x(t)) \right]^k
\end{array}
\end{equation}
 for some positive nubmers $\alpha$, $\chi$, $p$, $l$, $pk>1$, $lk<1$, then the set $\mathcal M \subset \mathbb R^n$ is globally fixed-time attractive for \eqref{eq_main} and 
 \begin{equation}
\label{eq_compare}
\begin{array}{lll}
T(x_0) \leq \frac{1}{\alpha^k(pk-1)} + \frac{1}{\chi^k(1-lk)}.
\end{array}
\end{equation} 
Here $\mathcal D^{+}$ denotes the upper right-hand derivative of the function $\phi(t)$: $\mathcal D^{+} \phi(t) := \limsup\limits_{h \to +0} \frac{\phi(t+h)-\phi(t)}{h}$.
\end{lemma}


Now we describe the main idea of the present paper and formulate three lemmas that will allow us to reduce the upper bound \eqref{eq_compare}. 
Then, in Remark \ref{Rem1} and Example 1, we discuss the results of these lemmas in detail, also comparing them with existing ones.


Consider the idea of this paper.
Denote by $x_i=x(t)|_{t=t_i}$, $i=0,1,2$, $t_0=0$.
Suppose that for the system \eqref{eq_main} it is possible to choose a radially unbounded positive-definite function $V: \mathbb R^n \to [0; +\infty)$ such that $\mathcal D^{+} V(x(t))<0$ for all $x$. 
Also, on the interval $V \in (V(x_1), V(x_2))$ we have $\mathcal D^{+} [\mathcal D^{+} V(x(t))] < 0$.
On the intervals $V \in (V(x_0), V(x_1)) \cup (V(x_2), +\infty)$ we have $\mathcal D^{+} [\mathcal D^{+} V(x(t))] > 0$, see Fig. \ref{Fig0}.
The presence of a function $V$ with upward convexity at $t \in (t_1,t_2)$ may indicate fast convergence of solutions \eqref{eq_main} to the equilibrium point.
Numerical examples will be considered below in this section.
Next, we will consider the application of this idea to the formulation of conditions for fast fixed-time convergence.

\begin{figure}[h]
\center{\includegraphics[width=0.8\linewidth]{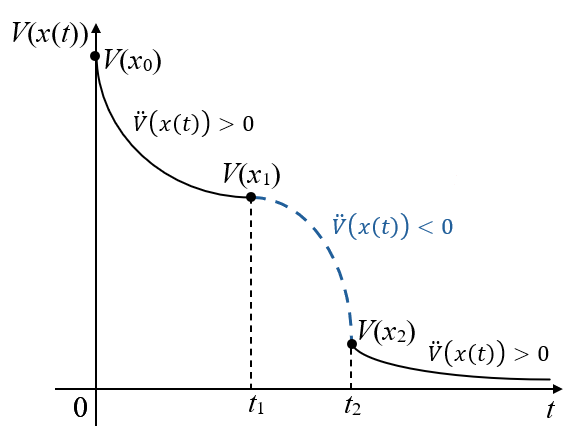}}
\caption{Example of alternating convexity upward and downward of the Lyapunov function $V(x(t))$.}
\label{Fig0}
\end{figure}


\begin{lemma}
\label{lem2}
If there exists a continuous radially unbounded positive definite function $V: \mathbb R^n \to [0; +\infty)$ such that
any solution $x(t)$ of the system \eqref{eq_main} satisfies the inequality
\begin{equation}
\label{eq_deriv_LF_main_Lem2}
\begin{array}{lll}
\mathcal D^{+} V(x(t)) \leq 
\\
-\left[ \alpha V^{p}(x(t)) + \beta V^{-q \cdot 1(1-V) +(q+r) \cdot 1(w-V)}(x(t)) \right]^k,
\end{array}
\end{equation}
where $\alpha$, $\beta$, $p$, $q$, $w$, $r$, $k$ are positive numbers, $pk>1$ and $rk<1$,
then the system \eqref{eq_main} is globally fixed-time stable, where the upper bound of settling-time function for any initial conditions $x(0)$ is
\begin{equation}
\label{eq_fix_time_Lem2}
\begin{array}{lll}
T_{\max} := \frac{1}{\alpha^k(pk-1)} + \frac{1-w^{1+qk}}{\beta^k(1+qk)} + \frac{w^{1-rk}}{\beta^k(1-rk)}.
\end{array}
\end{equation}
Here $1(c-V)$ is Heaviside function: $1(c-V) := \begin{cases}
0, & V <c,\\
1, & V \geq c.
\end{cases}$
\end{lemma}


\begin{remark}
It is clear from \eqref{eq_deriv_LF_main_Lem2} that the function $V(x(t))$ is monotonically decreasing and there exist coefficients $\alpha$, $\beta$, $p$, $q$, $w$, $k$ such that the alternation of upward and downward convexity of the function $V(x(t))$ corresponds to Fig. \ref{Fig0}.
\end{remark}


\begin{proof} 
In the proof we consider three cases: $V(x(t))>1$, $w < V(x(t)) \leq 1$ and $V(x(t)) \leq w$.


First, consider the case when $V(x(t)) >1$.
The expression \eqref{eq_deriv_LF_main_Lem2} can be estimated as
\begin{equation}
\label{eq_deriv1_LF_1_Lem2}
\begin{array}{lll}
\mathcal D^{+} V(x(t)) \leq - \alpha^{k} V^{pk}(x(t)).
\end{array}
\end{equation}
Find the solution of the inequality \eqref{eq_deriv1_LF_1_Lem2} in the form
\begin{equation}
\label{eq_deriv1_LF_2_Lem2}
\begin{array}{lll}
V(x(t)) \leq \left[ V^{1-pk}(x(0))- \alpha^{k}(1-pk)t \right]^{\frac{1}{1-pk}}.
\end{array}
\end{equation}
It follows from \eqref{eq_deriv1_LF_2_Lem2} that for any $x(t)$ such that $V(x(0))>1$ the condition $V(x(t)) \leq 1$ is satisfied at
\begin{equation}
\label{eq_deriv1_LF_2_time_Lem2}
\begin{array}{lll}
t \geq \frac{1}{\alpha^k(pk-1)} \left[ V^{1-pk}(x(0))- V^{1-pk}(x(t)) \right].
\end{array}
\end{equation}
Since $0 < V^{1-pk}(x(t)) \leq 1$ for any $V(x(t)) >1$ and $pk>1$, then the expression \eqref{eq_deriv1_LF_2_time_Lem2} can be estimated as
$t \geq \frac{1}{\alpha^k(pk-1)}$.


Consider the case when $w < V(x(t)) \leq 1$.
Inequality \eqref{eq_deriv_LF_main_Lem2} can be estimated as
\begin{equation}
\label{eq_deriv2_LF_1_Lem2}
\begin{array}{lll}
\mathcal D^{+} V(x(t)) \leq - \beta^k V^{-q k} (x(t)).
\end{array}
\end{equation}

For any $t_{01} \geq 0$ such that $w < V(x(t_{01})) \leq 1$, we find a solution of the inequality \eqref{eq_deriv2_LF_1_Lem2} in the form
\begin{equation}
\label{eq_deriv2_LF_2_Lem2_a}
\begin{array}{lll}
V(x(t)) \leq 
\\
\left[ V^{1+qk}(x(0)) - \beta^{k}(1+qk) (t-t_{01}) \right]^{\frac{1}{1+qk}}.
\end{array}
\end{equation}
It follows from \eqref{eq_deriv2_LF_2_Lem2_a} that for any $x(t)$ such that $w < V(x(t_{01})) \leq 1$ the inequality $w < V(x(t)) \leq 1$ is satisfied at
\begin{equation}
\label{eq_deriv2_LF_2_time_Lem2}
\begin{array}{lll}
t \geq t_{01} 
\\
+ \frac{1}{\beta^k(1+qk)} \left[ V^{1+qk}(x(t)) - V^{1+qk}(x(t_{01})) \right].
\end{array}
\end{equation}
Since $w < V(x(t)) \leq 1$ and $qk>0$, then the expression \eqref{eq_deriv2_LF_2_time_Lem2} can be rewritten as
$t \geq T_1 := \frac{1}{\alpha^k(pk-1)} + \frac{1-w^{1+qk}}{\beta^k(1+qk)}$.
Therefore, the value of $V(x(t))=w$ is achieved for any $t \geq T_{1}$.


Now consider the case when $V(x(t)) \leq w$.
Expression \eqref{eq_deriv_LF_main_Lem2} can be estimated as
\begin{equation}
\label{eq_deriv2_LF_1_Lem2_2}
\begin{array}{lll}
\mathcal D^{+} V(x(t)) \leq - \beta^k V^{r k} (x(t)).
\end{array}
\end{equation}

For any $t_{02} \geq 0$ such that $V(x(t_{02})) \leq w$, we find a solution of the inequality \eqref{eq_deriv2_LF_1_Lem2_2} in the form
\begin{equation}
\label{eq_deriv2_LF_2_Lem2}
\begin{array}{lll}
V(x(t)) \leq 
\\
\left[ V^{1-rk}(x(t_{02})) - \beta^{k}(1-rk) (t-t_{02}) \right]^{\frac{1}{1-rk}}.
\end{array}
\end{equation}
It follows from \eqref{eq_deriv2_LF_2_Lem2} that for any $x(t)$ such that $V(x(t_{02})) \leq w$ the value of $V(x(t)) =0$ is reached at
\begin{equation}
\label{eq_deriv2_LF_3_time_Lem2}
\begin{array}{lll}
t \geq t_{02} 
\\
+ \frac{1}{\beta^k(1-rk)} \left[ V^{1-rk}(x(t)) - V^{1-rk}(x(t_{02})) \right].
\end{array}
\end{equation}
Since $V(x(t)) \leq w$ and $rk<1$, then the expression \eqref{eq_deriv2_LF_3_time_Lem2} can be rewritten as
$t \geq \frac{1}{\alpha^k(pk-1)} + \frac{1-w^{1+qk}}{\beta^k(1+qk)} + \frac{w^{1-rk}}{\beta^k(1-rk)}$.
Therefore, the condition $x(t)=0$ is guaranteed to be satisfied at $t \geq T_{\max}$ for any initial conditions. 
Lemma \ref{lem2} is proved.

\end{proof}


\begin{lemma}
\label{lem1}
If there exists a continuous radially unbounded positive definite function $V=V(x(t)): \mathbb R^n \to [0; +\infty)$ such that
any solution $x(t)$ of the system \eqref{eq_main} satisfies the inequality
\begin{equation}
\label{eq_deriv_LF_main}
\begin{array}{lll}
\mathcal D^{+} V \leq
\begin{cases}
-\left[ \alpha V^{p} + \frac{\beta}{V^{q}} e^{-\frac{\gamma}{V^r}} + \chi V^{l}\right]^k, & V > 0,\\
0, & V = 0,
\end{cases}
\end{array}
\end{equation}
where $\alpha$, $\beta$, $\gamma$, $\chi$, $p$, $q$, $r$, $l$, $k$ are positive numbers, $pk>1$, and $lk<1$,
then the system \eqref{eq_main} is globally fixed-time stable, where the upper bound on settling-time function for any initial conditions $x(0)$ is
\begin{equation}
\label{eq_fix_time_1}
\begin{array}{lll}
T_{\max} := & \frac{1}{\alpha^k(pk-1)} + \frac{e^{\frac{k \gamma}{w^r}}(1-w^{1+qk})}{\beta^k(1+qk)}  
\\
&
+\frac{w^{1-lk}}{\chi^k(1-lk)}.
\end{array}
\end{equation}
Here $w \in (0;1]$.
\end{lemma}


\begin{remark}
It is clear from \eqref{eq_deriv_LF_main} that $\mathcal D^{+} V(x(t)) < 0$ for all $x \neq 0$.
Therefore, the function $V(x(t))$ is monotonically decreasing.
Moreover, there exist coefficients $\alpha$, $\beta$, $\gamma$, $\chi$, $p$, $q$, $r$, $l$, and $k$ such that the alternation of upward and downward convexity of the function $V(x(t))$ will correspond to Fig. \ref{Fig0}.
\end{remark}



\begin{proof}
In the proof we consider three cases: $V(x(t))>1$, $w < V(x(t)) \leq 1$ and $V(x(t)) \leq w$.
The proof of the case for $V(x(t)) >1$ is similar to the proof of lemma \ref{lem2}.




Consider the case when $w < V(x(t)) \leq 1$, where \eqref{eq_deriv_LF_main} can be estimated as
\begin{equation}
\label{eq_deriv2_LF_1}
\begin{array}{lll}
\mathcal D^{+} V(x(t)) \leq & - \frac{\beta^k}{V^{qk}(x(t))} e^{-\frac{k\gamma}{V^r(x(t))}} 
\\
& \leq - \frac{\beta^k}{V^{qk}(x(t))} e^{-\frac{k \gamma}{w^r}}.
\end{array}
\end{equation}

For any $t_{01} \geq 0$ such that $V(x(t_{01})) \leq 1$, we find a solution of the inequality \eqref{eq_deriv2_LF_1} in the form
\begin{equation}
\label{eq_deriv2_LF_2_lem2}
\begin{array}{lll}
V(x(t)) \leq 
\\
\left[ V^{1+qk}(x(t_{01}))- e^{-\frac{k \gamma}{w^r}} \beta^{k}(1+qk) (t-t_{01}) \right]^{\frac{1}{1+qk}}.
\end{array}
\end{equation}
It follows from \eqref{eq_deriv2_LF_2_lem2} that for any $x(t)$ such that $w < V(x(t_{01})) \leq 1$ the inequality $V(x(t)) \leq w$ is satisfied at
\begin{equation}
\label{eq_deriv2_LF_2_time_lem2}
\begin{array}{lll}
t \geq t_{01} + 
\\
\frac{e^{\frac{k \gamma}{w^r}}}{\beta^k(1+qk)} \left[ V^{1+qk}(x(t)) - V^{1+qk}(x(t_{01})) \right].
\end{array}
\end{equation}
Since $w < V(x(t)) \leq 1$ and $qk>0$, then the expression \eqref{eq_deriv2_LF_2_time_lem2} can be rewritten as
$t \geq T_1 := \frac{1}{\alpha^k(pk-1)} + \frac{e^{\frac{k \gamma}{w^r}}}{\beta^k(1+qk)}(1-w^{1+qk})$.
Therefore, the value of $V(x(t))=w$ is reached for any $t \geq T_1$.


Now consider the case when $V(x(t)) \leq w$.
Inequality \eqref{eq_deriv_LF_main} can be estimated as
\begin{equation}
\label{eq_deriv2_LF_1_2}
\begin{array}{lll}
\mathcal D^{+} V(x(t)) \leq - \chi^k V^{l k} (x(t)).
\end{array}
\end{equation}

For any $t_{02} \geq 0$ such that $V(x(t_{02})) \leq w$, we find a solution of the inequality \eqref{eq_deriv2_LF_1_2} in the form
\begin{equation}
\label{eq_deriv2_LF_2}
\begin{array}{lll}
V(x(t)) \leq 
\\
\left[ V^{1-lk}(x(t_{02})) - \chi^{k}(1-lk) (t-t_{02}) \right]^{\frac{1}{1-lk}}.
\end{array}
\end{equation}
It follows from \eqref{eq_deriv2_LF_2} that for any $x(t)$ such that $V(x(t_{02})) \leq w$ the value of $V(x(t)) =0$ is reached at
\begin{equation}
\label{eq_deriv2_LF_2_time}
\begin{array}{lll}
t \geq t_{02} + 
\\
\frac{1}{\chi^k(1-lk)} \left[ V^{1-lk}(x(t)) - V^{1-lk}(x(t_{02})) \right].
\end{array}
\end{equation}
Since $V(x(t)) \leq w$ and $lk<1$, then the expression \eqref{eq_deriv2_LF_2_time} can be rewritten as
$t \geq \frac{1}{\alpha^k(pk-1)} + \frac{e^{\frac{k \gamma}{w^r}}}{\beta^k(1+qk)}(1-w^{1+qk}) + \frac{w^{1-lk}}{\chi^k(1-lk)}$.
Therefore, the condition $x(t)=0$ is guaranteed to be satisfied at $t \geq T_{\max}$ for any initial conditions. 
Lemma \ref{lem1} is proved.

\end{proof}


\begin{lemma}
\label{lem3}
If there exists a continuous radially unbounded positive definite function $V=V(x(t)): \mathbb R^n \to [0; +\infty)$ such that
any solution $x(t)$ of the system \eqref{eq_main} satisfies the inequality
\begin{equation}
\label{eq_deriv_LF_main_lem3}
\begin{array}{lll}
\mathcal D^{+} V \leq
\begin{cases}
-\left[ \alpha V^{p} + \frac{\beta}{V^{q}} e^{-\frac{\gamma}{V^r}} \right]^k, & V > 0,\\
0, & V = 0,
\end{cases}
\end{array}
\end{equation}
where $\alpha$, $\beta$, $\gamma$, $p$, $q$, $r$, $k$ are positive numbers and $pk>1$,
then the set
\begin{equation}
\label{eq_set_M_lem3}
\begin{array}{lll}
\mathcal M = \left\{x \in \mathbb R^n: |x| \leq V^{-1} (w), w \in (0, 1) \right\}
\end{array}
\end{equation}
is globally attractive with an upper bound on the convergence time
\begin{equation}
\label{eq_fix_time_1_lem3}
\begin{array}{lll}
T_{\max} := \frac{1}{\alpha^k(pk-1)} + \frac{e^{\frac{k \gamma}{w^r}}(1-w^{1+qk})}{\beta^k(1+qk)}
\end{array}
\end{equation}
for any initial conditions $x(0)$, and $\lim\limits_{t \to \infty} x(t) = 0$.
\end{lemma}


\begin{proof}
In the proof we consider three cases: $V(x(t))>1$, $w < V(x(t)) \leq 1$ and $V(x(t)) \leq w$, where the proof of the cases for $V(x(t)) >1$ and $w < V(x(t)) \leq 1$ is similar to the proof of lemma \ref{lem1}. 
Also, the value of \eqref{eq_fix_time_1_lem3} is obtained in the proof of lemma \ref{lem1}.


Now consider the domain $V(x(t)) \leq w$.
It follows from the conditions $V: \mathbb R^n \to [0;+\infty)$ and \eqref{eq_deriv_LF_main_lem3} that the function $V$ is monotonically decreasing.
Consequently, the solutions \eqref{eq_deriv_LF_main_lem3}, once having entered the domain \eqref{eq_set_M_lem3}, never leave it.
Additionally, applying L'H$\hat{\textup{o}}$pital's (also known as Bernoulli's) rule several times for the term $\frac{\beta}{V^{q}(x(t))} e^{-\frac{\gamma}{V^r(x(t))}}$ in \eqref{eq_deriv_LF_main_lem3}, one has $\lim\limits_{|x| \to 0}\mathcal D^{+} V(x(t)) =0$. 
Hence, $\lim\limits_{t \to \infty}x(t)=0$.
Lemma \ref{lem3} is proved.

\end{proof}


\begin{remark}
\label{Rem1}
Let us consider the difference between the proposed result and \cite{Polyakov12}.
In \cite{Polyakov12}, the convergence time estimate is \eqref{eq_compare}.
In \eqref{eq_compare}, the value of the second term is limited by the choice of $l$ and $k$ so that $lk<1$.
This indicates a limitation on the convergence rate in the region $V(x(t))<1$ for a fixed $\chi$.
Compared with \eqref{eq_compare}, the estimate \eqref{eq_fix_time_1} is reduced by the second term, where the denominator can be made sufficiently large by appropriately choosing $q$ and $k$ for a fixed $\beta$.
Moreover, by increasing $p$, $k$ and $r$, the estimate $T_{\max}$ can be made arbitrarily small.

Unlike \eqref{eq_compare} the estimate \eqref{eq_fix_time_Lem2} is reduced by the second term, as in \eqref{eq_fix_time_1}.
Also, compared with the second term in \eqref{eq_compare}, the numerator of the third term in \eqref{eq_fix_time_Lem2} contains the value $w$, which can be chosen sufficiently small.
\end{remark}




\textit{Example 1.} 
Consider a first-order integrator
\begin{equation}
\label{eq_plant}
\begin{array}{lll}
\dot{x}(t)=u(t),
\end{array}
\end{equation}
where $x \in \mathbb R$, $u \in \mathbb R$ is a control.

To illustrate the obtained results, we consider seven different control laws.
The linear control law
\begin{equation}
\label{eq_CL_1}
\begin{array}{lll}
u=-x
\end{array}
\end{equation}
guarantees exponential convergence to zero of the solutions of the closed-loop system \eqref{eq_plant}, \eqref{eq_CL_1}. 
By choosing $V=0.5x^2$ and taking into account \eqref{eq_plant}, \eqref{eq_CL_1}, we obtain $\dot{V}=-x^2=-2V<0$ and $\ddot{V}=4V>0$ for all $x \neq 0$.
Thus, the function $V$ does not change the direction of convexity downwards for all $x$.

According to \cite{Utkin92}, the relay control law
\begin{equation}
\label{eq_CL_2}
\begin{array}{lll}
u=- sign\{x\}
\end{array}
\end{equation}
guarantees the convergence of the solutions of the closed-loop system \eqref{eq_plant}, \eqref{eq_CL_2} to zero in a finite-time. 
By choosing $V=0.5x^2$ and considering \eqref{eq_plant}, \eqref{eq_CL_2}, one has $\dot{V}=-|x|=-\sqrt{2}V^{0.5}<0$ and $\ddot{V}=1>0$ for all $x \neq 0$.
Thus, the function $V$ does not change the direction of convexity downwards for all $x$.

Following \cite{Haimo86,Emelyanov86,Bhat00,Davila05,Moulay06,Orlov09,Bejarano10,Orlov11,Efimov11,Utkin13}, the continuous control law
\begin{equation}
\label{eq_CL_3}
\begin{array}{lll}
u=- x^{1/3}
\end{array}
\end{equation}
guarantees the convergence of the solutions of the closed-loop system \eqref{eq_plant}, \eqref{eq_CL_3} to zero in finite-time. 
By choosing $V=0.5x^2$ and taking into account \eqref{eq_plant}, \eqref{eq_CL_3}, we have $\dot{V}=-|x|^{4/3}=-2^{2/3}V^{2/3}<0$ and $\ddot{V}=\frac{2^{7/3}}{3}V^{1/3}>0$ for all $x \neq 0$.
Thus, the function $V$ does not change the direction of convexity downwards for all $x$.

According to \cite{Polyakov12}, the continuous control law
\begin{equation}
\label{eq_CL_4}
\begin{array}{lll}
u=- x^3-x^{1/3}
\end{array}
\end{equation}
guarantees the convergence of the solutions of the closed-loop system \eqref{eq_plant}, \eqref{eq_CL_4} to zero in a fixed-time $T_{\max}=2$ according to \eqref{eq_compare}. 
By choosing $V=0.5x^2$ and considering \eqref{eq_plant}, \eqref{eq_CL_4}, we obtain $\dot{V}=-x^4-x^{4/3}=-4V^2-2^{2/3}V^{2/3}<0$ and $\ddot{V}=(-8V-\frac{2^{5/3}}{3}V^{-1/3})\dot{V}>0$ for all $x \neq 0$.
Thus, the function $V$ does not change the direction of convexity downwards for all $x$.

Following \cite{Nekhoroshikh23}, the control law
\begin{equation}
\label{eq_CL_5}
u= \begin{cases}
-x(1+|\ln|x||), & x \neq 0,\\
0, & x=0
\end{cases}
\end{equation}
guarantees hyperexponential convergence of solutions of the closed-loop system \eqref{eq_plant}, \eqref{eq_CL_5} to zero as $t \to +\infty$. 
By choosing $V=0.5x^2$ and taking into account \eqref{eq_plant}, \eqref{eq_CL_5}, one gets $\dot{V}=-2V(1+|\ln(\sqrt{2}V^{0.5})|)<0$ and $\ddot{V}=-2(1+0.5sign\{\ln(\sqrt{2}V^{0.5})\}+|\ln(\sqrt{2}V^{0.5})|)\dot{V}>0$ for all $x \neq 0$.
Thus, the function $V$ does not change the direction of convexity downwards for all $x$.

The control law
\begin{equation}
\label{eq_CL_7}
\begin{array}{lll}
u=-x^{3}-|x|^{-3 \cdot 1(1-|x|) + (3+1/3) \cdot 1(0.01-|x|)} sign\{x\}
\end{array}
\end{equation}
satisfies the condition of Lemma \ref{lem2} with $V=0.5x^2$, i.e., guarantees the convergence of the solutions of the closed-loop system \eqref{eq_plant}, \eqref{eq_CL_7} to zero at a fixed-time $T_{\max}=0.82$ according to \eqref{eq_fix_time_Lem2} with $w=0.01$. 
Taking into account \eqref{eq_plant} and \eqref{eq_CL_7}, we obtain $\dot{V}=-x^4-|x|^{-3 \cdot 1(1-|x|) + (3+1/3) \cdot 1(0.01-|x|)+1}$.
Let us consider three intervals.
On the first interval $x > 1$ we have $\dot{V}=-x^4-|x|=-4V^2-\sqrt{2} V^{0.5} < 0$ and $\ddot{V}=(-8V-2^{-0.5}V^{-0.5}) \dot{V} > 0$.
On the second interval $x \in (0.01,1)$ we have $\dot{V}=-x^4-|x|^{-2}=-4V^2-0.5V^{-1}<0$ and $\ddot{V}=(-8V+\frac{1}{2V^2})\dot{V}$, where $\ddot{V}>0$ for $x \in (0.01,2^{-1/6})$ and $\ddot{V}<0$ for $x \in (2^{-1/6},1)$.
On the third interval $x \in (0,0.01)$ we have $\dot{V}=-x^4-|x|^{4/3}=-4V^2-2^{2/3}V^{2/3}<0$ and $\ddot{V}=(-8V-\frac{2^{5/3}}{3}V^{-1/3})\dot{V}>0$.
Thus, the function $V$ changes the direction of convexity, as in Fig. \ref{Fig0}.

The control law
\begin{equation}
\label{eq_CL_6}
\begin{array}{lll}
u=
\begin{cases}
-x^{3}-\frac{1}{x^{3}}e^{-\frac{0.05}{|x|^{0.5}}}-x^{1/3}, & x \neq 0, \\
0, & x = 0
\end{cases}
\end{array}
\end{equation}
satisfies the condition of Lemma \ref{lem1} with $V=0.5x^2$, i.e. guarantees the convergence of solutions of the closed-loop system \eqref{eq_plant}, \eqref{eq_CL_6} to zero at a fixed-time $T_{\max}=0.97$ with $w=0.02$. 
Taking into account \eqref{eq_plant} and \eqref{eq_CL_6}, we obtain $\dot{V}=-x^4-\frac{1}{x^2}e^{-\frac{0.05}{|x|^{0.5}}}-x^{4/3}=
-4V^2-\frac{1}{2V}e^{-\frac{0.05}{2^{0.25}V^{0.25}}}-\frac{2}{3}V^{2/3}$ and $\ddot{V}=(-8V+\frac{1}{2V^2}e^{-\frac{0.05}{2^{0.25}V^{0.25}}}-\frac{1}{160 2^0.25}V^{9/4}e^{-\frac{0.05}{2^{0.25}V^{0.25}}}-2^{2/3}\frac{2}{3}V^{-1/3})\dot{V}$.
We have $\dot{V}<0$ for any $x \neq 0$.
On the intervals $V \in (0, 1.22 \cdot 10^{-8}) \cup (0.33, +\infty)$ we have $\ddot{V}>0$ and on the interval $V \in (1.22 \cdot 10^{-8},0.33)$ we have $\ddot{V}<0$.
Thus, the function $V$ changes the direction of convexity, as in Fig. \ref{Fig0}.

The control law
\begin{equation}
\label{eq_CL_8}
\begin{array}{lll}
u=
\begin{cases}
-x^{3}-\frac{1}{x^{3}}e^{-\frac{0.05}{|x|^{0.5}}}, & x \neq 0, \\
0, & x = 0
\end{cases}
\end{array}
\end{equation}
satisfies the condition of Lemma \ref{lem3} with $V=0.5x^2$, i.e. guarantees the asymptotic convergence of solutions of the closed-loop system \eqref{eq_plant}, \eqref{eq_CL_8} to zero, as well as the entry of solutions into the domain $\mathcal M = \left\{x \in \mathbb R: |x| \leq w:=0.01 \right\}$ in a fixed-time $T_{\max}=0.91$ according to \eqref{eq_fix_time_1_lem3}. 
Taking into account \eqref{eq_plant} and \eqref{eq_CL_6}, we obtain $\dot{V}=-4V^2-\frac{1}{2V}e^{-\frac{0.5}{2^{0.25}V^{0.25}}}<0$ (for any $x \neq 0$) and
$\ddot{V}= \Big( 8V + \frac{1}{2V^2}e^{-\frac{0.5}{2^{0.25}V^{0.25}}} - \frac{1}{160 2^{0.25}V^{9/4}}e^{-\frac{0.5}{2^{0.25}V^{0.25}}} \Big)\dot{V}$, where $\ddot{V}>0$ on intervals $V \in (0,1.22 \cdot 10^{-8}) \cup (0.39, +\infty)$ and $\ddot{V}<0$ on the interval $V \in (1.22 \cdot 10^{-8}, 0.39)$.
Thus, the function $V$ changes the direction of convexity, as in Fig. \ref{Fig0}.

In Fig.~\ref{Fig_1_solutions_a}, the graphs of solutions using the control laws \eqref{eq_CL_1}-\eqref{eq_CL_8} for $x(0)=3$ are shown.
From Fig.~\ref{Fig_1_solutions_a}, it is evident that in the region $x<1$, that the convergence towards the equilibrium point of the proposed solutions \eqref{eq_CL_7}-\eqref{eq_CL_8} is faster compared to the control laws \eqref{eq_CL_1}-\eqref{eq_CL_5}.
In Fig.~\ref{Fig_1_solutions_b}, the convergence of solutions \eqref{eq_plant}, \eqref{eq_CL_7}-\eqref{eq_CL_8} in a fixed-time under initial conditions  $x(0)=1$, $x(0)=10$, $x(0)=10^2$, and $x(0)=10^3$ is demonstrated.


\begin{figure}[h]
\center{\includegraphics[width=0.8\linewidth]{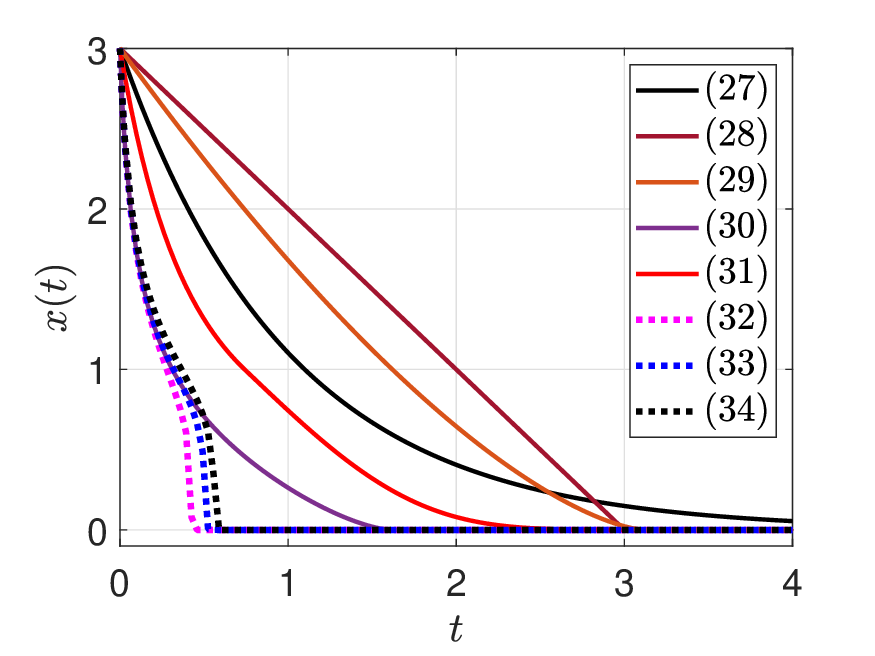}}
\caption{The solutions of \eqref{eq_plant} using the control laws \eqref{eq_CL_1}-\eqref{eq_CL_8} for $x(0)=3$.}
\label{Fig_1_solutions_a}
\end{figure}

\begin{figure}[h]
\center{\includegraphics[width=0.8\linewidth]{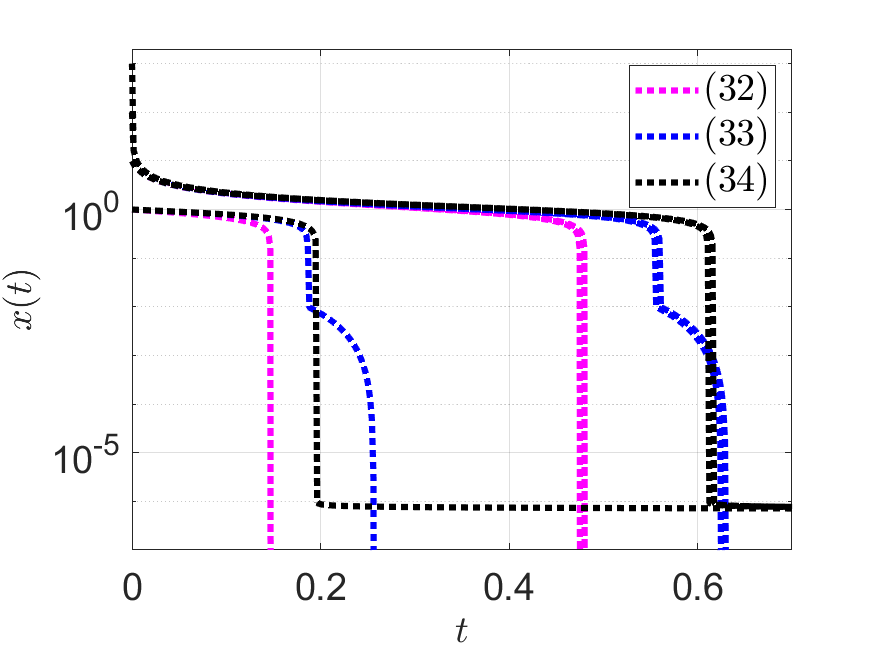}}
\caption{The solutions of \eqref{eq_plant} with the control laws \eqref{eq_CL_7}-\eqref{eq_CL_8} under initial conditions $x(0)=1$, $x(0)=10$, $x(0)=10^2$, and $x(0)=10^3$.}
\label{Fig_1_solutions_b}
\end{figure}


\section{Control Law Design}
\label{Sec_Control}

Consider the linear plant model in the form
\begin{equation}
\label{eq_lin_plant}
\begin{array}{lll}
\dot{x}_i=x_{i+1},~i=1,...,n-1,
\\
\dot{x}_n=\sum_{j=1}^{n} a_j x_j +u,
\end{array}
\end{equation}
where $x_i \in \mathbb R$,
$u \in \mathbb R$ is a control,
$a_j$ are known numbers.
For the system \eqref{eq_lin_plant} the controllability conditions are satisfied.

The goal is to stabilize the system \eqref{eq_lin_plant} using the results of lemmas \ref{lem1} and \ref{lem3}.
To design the control law, we use the backstepping method, see \cite{Khalil00}.
Divide the control law design into $n$ steps.


\textit{Step 1.} Introduce the first virtual control law for the first equation in \eqref{eq_lin_plant} in the form
\begin{equation}
\label{eq_lin_CL_1}
\begin{array}{lll}
v_1=x_{2}=\varphi_1(x_1),
\end{array}
\end{equation}
where $\varphi_1(x_1)=
\begin{cases}
-\alpha_1 x_1^{p_1}-\frac{\beta_1}{x_1^{q_1}} e^{-\frac{\gamma_1}{x_1^{r_1}}}, & |x_1|>0, \\
0, & x_1=0,
\end{cases}
$ 
$\alpha_1>0$, $\beta_1>0$, $\gamma_1 >0$, the numbers $p_1$ and $q_1$ are odd, the number $r_1>0$ is such that the function $x_1^{r_1}$ is even.


%
%


\textit{Step $i$, $i=2,...,n-1$.}
Consider the error $e_{i-1}=x_{i}-v_{i-1}$.
Differentiating it w.r.t. time, one gets
\begin{equation}
\label{eq_lin_CL_er_i}
\begin{array}{lll}
\dot{e}_{i-1} = & x_{i+1} - \dot{\varphi}_{i-1}(e_{i-2}) - \ddot{\varphi}_{i-2}(e_{i-3}) - ...
\\
& - \varphi^{(i-2)}_2(e_1) - \varphi^{(i-1)}_1(x_1), ~~~ i=2,...,n-1.
\end{array}
\end{equation}

Introduce the virtual control laws for the equations \eqref{eq_lin_CL_er_i} in the forms
\begin{equation}
\label{eq_lin_CL_2}
\begin{array}{lll}
v_i = & x_{i+1}=\varphi_{i}(e_{i-1}) + \dot{\varphi}_{i-1}(e_{i-2}) + \ddot{\varphi}_{i-2}(e_{i-3}) 
\\
& ... + \varphi^{(i-2)}_2(e_1) + \varphi^{(i-1)}_1(x_1), 
\\
& i=2,...,n-1,
\end{array}
\end{equation}
where $\varphi_i(e_{i-1})=
\begin{cases}
-\alpha_i e_{i-1}^{p_i}-\frac{\beta_i}{e_{i-1}^{q_i}} e^{-\frac{\gamma_i}{e_{i-1}^{r_i}}}, & |e_{i-1}|>0, \\
0, & e_{i-1}=0,
\end{cases}$
$\alpha_i>0$, $\beta_i>0$, $\gamma_i >0$, numbers $p_i$ and $q_i$ are odd, number $r_i>0$ is such that function $e_{i-1}^{r_i}$ is even.


\textit{Step $n$.} Consider the error $e_{n-1}=x_{n}-v_{n-1}$.
Differentiating it w.r.t. time, one has
\begin{equation}
\label{eq_lin_CL_er_n}
\begin{array}{lll}
\dot{e}_{n-1}= & u - \dot{\varphi}_{n-1}(e_{n-2}) 
\\
& ... - \varphi^{(n-2)}_{2}(e_1) - \varphi^{(n-1)}_1(x_1).
\end{array}
\end{equation}

Introduce the control law for the equation \eqref{eq_lin_CL_er_n} in the form
\begin{equation}
\label{eq_lin_CL_n}
\begin{array}{lll}
u= & \varphi_n(e_{n-1}) + \dot{\varphi}_{n-1}(e_{n-2}) 
\\
&... + \varphi^{(n-2)}_2(e_1) + \varphi^{(n-1)}_1(x_1) - \sum_{j=1}^{n} a_j x_j,
\end{array}
\end{equation}
where 
\begin{equation*}
\begin{array}{lll}
\varphi_n(e_{n-1})=
\begin{cases}
-\alpha_n e_{n-1}^{p_n}-\frac{\beta_n}{e_{n-1}^{q_n}} e^{-\frac{\gamma_n}{e_{n-1}^{r_n}}} 
\\
- \chi |e_{n-1}|^{l} sign\{e_{n-1}\}, & |e_{n-1}|>0, \\
0, & e_{n-1}=0,
\end{cases}
\end{array}
\end{equation*}
$\alpha_n>0$, $\beta_n>0$, $\gamma_n >0$, $\chi>0$, $0<l<1$, the numbers $p_n$ and $q_n$ are odd, the number $r_n>0$ is such that the function $e_{n-1}^{r_n}$ is even.


\begin{theorem}
\label{Th1}
The virtual control laws \eqref{eq_lin_CL_1}, \eqref{eq_lin_CL_2} together with the control law \eqref{eq_lin_CL_n} guarantee the convergence of $x_1$, $e_i$, $i=1,...,n-2$ in the domains
\begin{equation}
\label{eq_set_M_all}
\begin{array}{lll}
\mathcal M_{x_1} = & \left\{x_1 \in \mathbb R: |x_1| \leq w_0, w_0 \in (0,1) \right\},
\\
\mathcal M_{e_i} = & \left\{e_i \in \mathbb R: |e_i| \leq w_i, w_i \in (0,1) \right\}, 
\\
& i=1,...,n-2,
\end{array}
\end{equation}
in fixed-time 
\begin{equation}
\label{eq_set_T_all}
\begin{array}{lll}
T_{\max}^{x_1} = \frac{\sqrt{2}}{\alpha_1 (p_1+1)}+\frac{2\sqrt{2} e^{\frac{\sqrt{2} \gamma_1}{w^{r_1}}}(1-w^{0.5(1+q_1)})}{\beta_1 (1+q_1)},
\\
T_{\max}^{e_i} = \frac{\sqrt{2}}{\alpha_{i+1} (p_{i+1}+1)}
\\
+\frac{2\sqrt{2} e^{\frac{\sqrt{2} \gamma_{i+1}}{w^{r_{i+1}}}}(1-w^{0.5(1+q_{i+1})})}{\beta_1 (1+q_{i+1})}, ~~~ i=1,...,n-2,
\end{array}
\end{equation}
correspondingly and asymptotic convergence of $x_i$, $i=1,...,n$ to zero. 
Also, $e_{n-1} = 0 $ after 
$t \geq T_{\max}^{e_{n-1}}:= \frac{\sqrt{2}}{\alpha_{n} (p_{n}+1)}+\frac{2\sqrt{2} e^{\frac{\sqrt{2} \gamma_{n}}{w^{r_{n}}}}(1-w^{0.5(1+q_{n})})}{\beta_1 (1+q_{n})}+\frac{\sqrt{2} w^{0.5(1-l)}}{\chi (1-l)}$.

\end{theorem}


\begin{proof} Divide the proof into $n$ steps.


\textit{Step 1.} Substituting \eqref{eq_lin_CL_1} into the first equation \eqref{eq_lin_plant}, we obtain
\begin{equation}
\label{eq_lin_CLS_1}
\begin{array}{lll}
\dot{x}_1 = \varphi_1(x_1).
\end{array}
\end{equation}

To analyze the stability of \eqref{eq_lin_CLS_1}, introduce the Lyapunov function
\begin{equation}
\label{eq_lin_FL_1}
\begin{array}{lll}
V_0=0.5 x_1^2.
\end{array}
\end{equation}
Differentiating \eqref{eq_lin_FL_1} w.r.t. time along the trajectories of \eqref{eq_lin_CLS_1}, one has
\begin{equation}
\label{eq_lin_dot_FL_1}
\begin{array}{lll}
\dot{V}_0= -\alpha_1 x_1^{p_1+1}-\frac{\beta_1}{x_1^{q_1-1}} e^{-\frac{\gamma_1}{x_1^{r_1}}}.
\end{array}
\end{equation}

It follows from \eqref{eq_lin_FL_1} that $|x_1|=\sqrt{2 V_1}$.
Then \eqref{eq_lin_dot_FL_1} can be rewritten as
\begin{equation*}
\label{eq_lin_dot_FL_12}
\begin{array}{lll}
\dot{V}_0= -\alpha_1 \sqrt{2} V_0^{\frac{p_1+1}{2}}-\frac{\beta_1}{\sqrt{2} V_0^{\frac{q_1-1}{2}}} e^{-\frac{\gamma_1}{\sqrt{2} V_0^{0.5r_1}}}.
\end{array}
\end{equation*}
Thus, the solutions $x_1(t)$ satisfy the condition of the lemma \ref{lem1}.
Hence, the trajectories $x_1(t)$ converge to the region $\mathcal M_{x_1}$ in a fixed-time (see the first expression in \eqref{eq_set_T_all}) and asymptotically tend to zero.
Additionally, $\lim\limits_{t \to +\infty} v_1(t) = 0$ from \eqref{eq_lin_CL_1}.


%
%


\textit{Step $i=2,...,n-1$.} Substituting \eqref{eq_lin_CL_2} into \eqref{eq_lin_CL_er_i}, we obtain
\begin{equation}
\label{eq_lin_CLS_2}
\begin{array}{lll}
\dot{e}_{i-1}=\varphi_{i}(e_{i-1}).
\end{array}
\end{equation}

Introduce the Lyapunov function
\begin{equation}
\label{eq_lin_FL_2}
\begin{array}{lll}
V_{i-1}=0.5 e_{i-1}^2.
\end{array}
\end{equation}
Differentiating \eqref{eq_lin_FL_2} w.r.t. time along the trajectories \eqref{eq_lin_CLS_2}, one has
\begin{equation}
\label{eq_lin_dot_FL_2}
\begin{array}{lll}
\dot{V}_{i-1}= -\alpha_i e_{i-1}^{p_{i}+1}-\frac{\beta_{i}}{e_{i-1}^{q_{i}-1}} e^{-\frac{\gamma_{i}}{e_{i-1}^{r_{i}}}}.
\end{array}
\end{equation}

It follows from \eqref{eq_lin_FL_2} that $|e_{i-1}|=\sqrt{2 V_{i-1}}$.
Then \eqref{eq_lin_dot_FL_2} can be rewritten as
\begin{equation*}
\label{eq_lin_dot_FL_22}
\begin{array}{lll}
\dot{V}_{i-1}= - \sqrt{2} \alpha_{i} V_{i-1}^{\frac{p_i+1}{2}}-\frac{\beta_i}{\sqrt{2} V_{i-1}^{\frac{q_i-1}{2}}} e^{-\frac{\gamma_i}{\sqrt{2} V_{i-1}^{0.5r_2}}}.
\end{array}
\end{equation*}
Thus, the solutions $e_{i-1}(t)$ satisfy the condition of lemma \ref{lem1}.
Hence, the trajectories $e_{i-1}(t)$ converge to the domain $\mathcal M_{e_{i-1}}$ in a fixed-time (see \eqref{eq_set_T_all}) and asymptotically tend to zero.
It follows from the structures of \eqref{eq_lin_CL_1}, \eqref{eq_lin_CL_2} that $\lim\limits_{t \to +\infty} x_i(t) = 0$, $\lim\limits_{t \to +\infty} v_j(t) = 0$, $j=2,...,n-1$.


\textit{Step $n$.} Substituting \eqref{eq_lin_CL_n} into \eqref{eq_lin_CL_er_n}, we obtain
\begin{equation}
\label{eq_lin_CLS_2_n}
\begin{array}{lll}
\dot{e}_{n-1}=\varphi_{n}(e_{n-1}).
\end{array}
\end{equation}

Consider the Lyapunov function
\begin{equation}
\label{eq_lin_FL_2_n}
\begin{array}{lll}
V_{n-1}=0.5 e_{n-1}^2.
\end{array}
\end{equation}
Differentiating \eqref{eq_lin_FL_2_n} w.r.t. time along the trajectories \eqref{eq_lin_CLS_2_n}, we have
\begin{equation}
\label{eq_lin_dot_FL_2_n}
\begin{array}{lll}
\dot{V}_{n-1}= -\alpha_n e_{n-1}^{p_{n}+1}-\frac{\beta_{n}}{e_{n-1}^{q_{n}-1}} e^{-\frac{\gamma_{n}}{e_{n-1}^{r_{n}}}} 
- \chi |e_{n-1}|^{l+1}.
\end{array}
\end{equation}

It follows from \eqref{eq_lin_FL_2_n} that $|e_{n-1}|=\sqrt{2 V_{n-1}}$.
Then \eqref{eq_lin_dot_FL_2_n} can be rewritten as
\begin{equation*}
\label{eq_lin_dot_FL_22_n}
\begin{array}{lll}
\dot{V}_{n-1} = & -\alpha_{n} \sqrt{2} V_{n-1}^{\frac{p_n+1}{2}}
\\
&-\frac{\beta_n}{\sqrt{2} V_{n-1}^{\frac{q_n-1}{2}}} e^{-\frac{\gamma_n}{\sqrt{2} V_{n-1}^{0.5r_n}}} - \sqrt{2} \chi V_{n-1}^{0.5(l+1)}.
\end{array}
\end{equation*}
Thus, the solutions $e_{n-1}(t)$ satisfy the condition of lemma \ref{lem3}.
Hence, the trajectories $e_{n-1}(t)$ converge to zero at a fixed-time $T_{\max}^{e_{n-1}}$.
It follows from the structure of the control law \eqref{eq_lin_CL_n}  that $\lim\limits_{t \to +\infty} x_n(t) = 0$ and $\lim\limits_{t \to +\infty} u(t) = 0$.
Theorem \ref{Th1} is proved.

\end{proof}


\textit{Example 2.}
Consider a second-order integrator in the form
\begin{equation}
\label{eq_lin_plant_ex1}
\begin{array}{lll}
\dot{x}_1=x_{2},
\\
\dot{x}_2=u.
\end{array}
\end{equation}

According to \eqref{eq_lin_CL_1} and \eqref{eq_lin_CL_n}, introduce the control law as follows
\begin{equation}
\label{eq_lin_plant_CL_ex11}
\begin{array}{lll}
u = & \begin{cases}
-\alpha_2 e_1^{p_2}-\frac{\beta_2}{e_1^{q_2}} e^{-\frac{\gamma_2}{e_1^{r_2}}}
\\
+ \dot{\varphi}_1(x_1) - \chi |e_1|^l sign\{e_1\}, & |e_1|>0, \\
0, & e_1=0,
\end{cases}
\\
e_1= & x_2-v_1,
\\
v_1 = & \begin{cases}
-\alpha_1 x_1^{p_1}-\frac{\beta_1}{x_1^{q_1}} e^{-\frac{\gamma_1}{x_1^{r_1}}}, & |x_1|>0, \\
0, & x_1=0,
\end{cases}
\\
\dot{\varphi}_1(x_1) = & \begin{cases}
-x_2 \Big( 
\alpha_1 p_1 x_1^{p_1-1} 
\\
- \frac{\beta_1 q_1}{x_1^{q_1+1}} e^{-\frac{\gamma_1}{x_1^{r_1}}} 
+ \frac{\beta_1 \gamma_1 r_1}{x_1^{q_1+r_1+1}} e^{-\frac{\gamma_1}{x_1^{r_1}}}
\Big), & |x_1|>0, \\
0, & x_1=0.
\end{cases}
\end{array}
\end{equation}

For comparison with the proposed control law \eqref{eq_lin_plant_CL_ex11} we introduce a linear control law
\begin{equation}
\label{eq_lin_plant_CL_ex12}
\begin{array}{lll}
u = Kx.
\end{array}
\end{equation}

Choose $\alpha_1=\alpha_2=\beta_1=\beta_2=\gamma_1=\gamma_2=\chi=1$, $p_1=p_2=q_1=q_2=3$, $r_1=r_2=0.8$, $l=1/3$ in \eqref{eq_lin_plant_CL_ex11} and $K=[-4~~-4]$ in \eqref{eq_lin_plant_CL_ex12}.

Fig.~\ref{Fig_21_control} shows the graphs of solutions of \eqref{eq_lin_plant_ex1} using the control laws \eqref{eq_lin_plant_CL_ex11}, \eqref{eq_lin_plant_CL_ex12} for $x(0)=col\{2,2\}$.
Fig.~\ref{Fig_22_control} shows the graphs of the solutions for the initial conditions $x(0)=col\{20,20\}$ increased by 10 times.
Fig.~\ref{Fig_23_control} shows the graphs of the solutions for $x(0)=col\{2,2\}$ and an additive disturbance in the second equation \eqref{eq_lin_plant_ex1} in the form $\dot{x}_2=u+20 \sin(t)$.
It is evident from the figures that the proposed control law are more robust with respect to the initial conditions and external disturbances than a linear control law.

\begin{figure}[h]
\center{\includegraphics[width=0.8\linewidth]{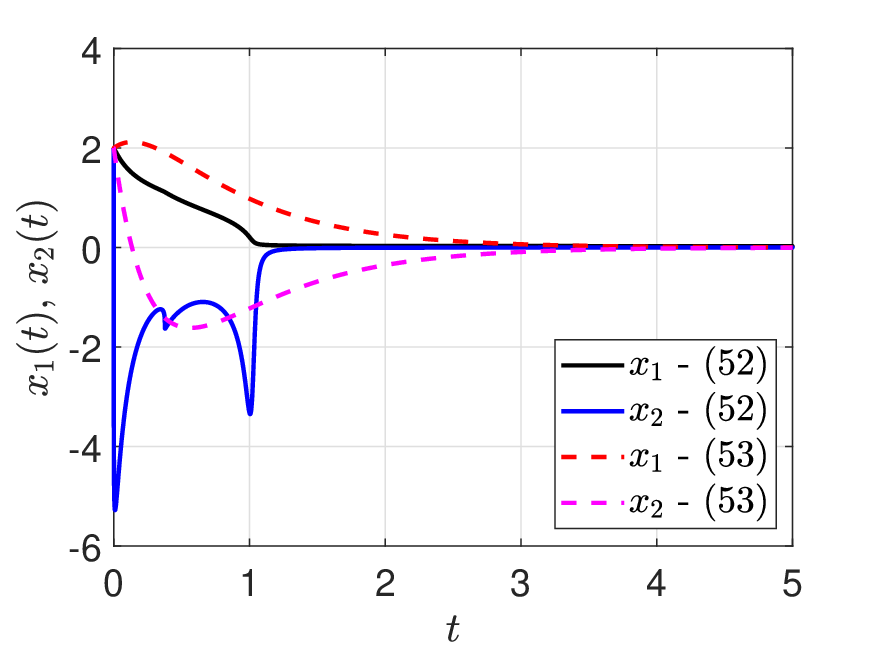}}
\caption{The graphs of solutions of \eqref{eq_lin_plant_ex1} using control laws \eqref{eq_lin_plant_CL_ex11} and \eqref{eq_lin_plant_CL_ex12} for $x(0)=col\{2,2\}$.}
\label{Fig_21_control}
\end{figure}

\begin{figure}[h]
\center{\includegraphics[width=0.8\linewidth]{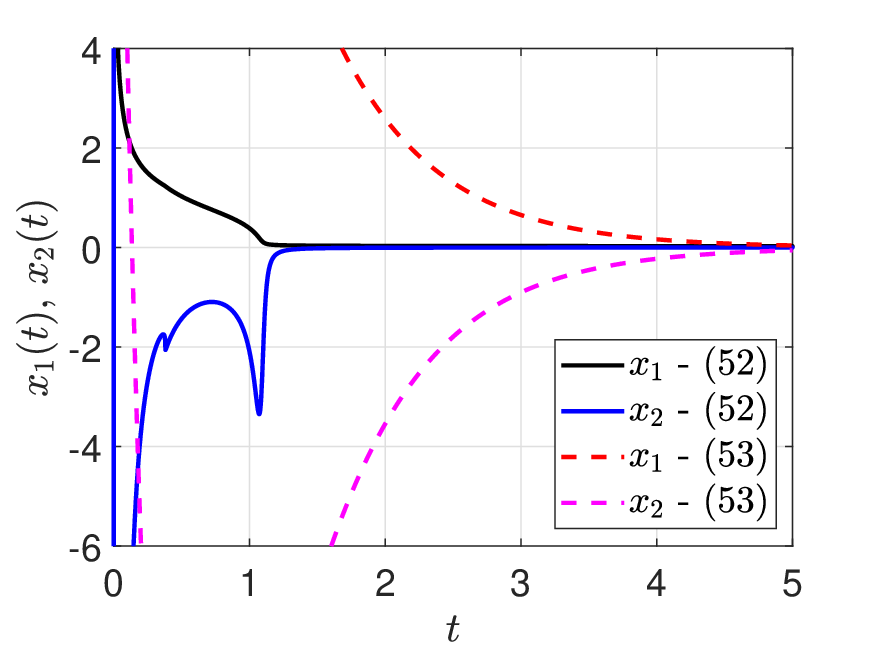}}
\caption{The graphs of solutions of \eqref{eq_lin_plant_ex1} using control laws \eqref{eq_lin_plant_CL_ex11} and \eqref{eq_lin_plant_CL_ex12} for $x(0)=col\{20,20\}$.}
\label{Fig_22_control}
\end{figure}

\begin{figure}[h]
\center{\includegraphics[width=0.8\linewidth]{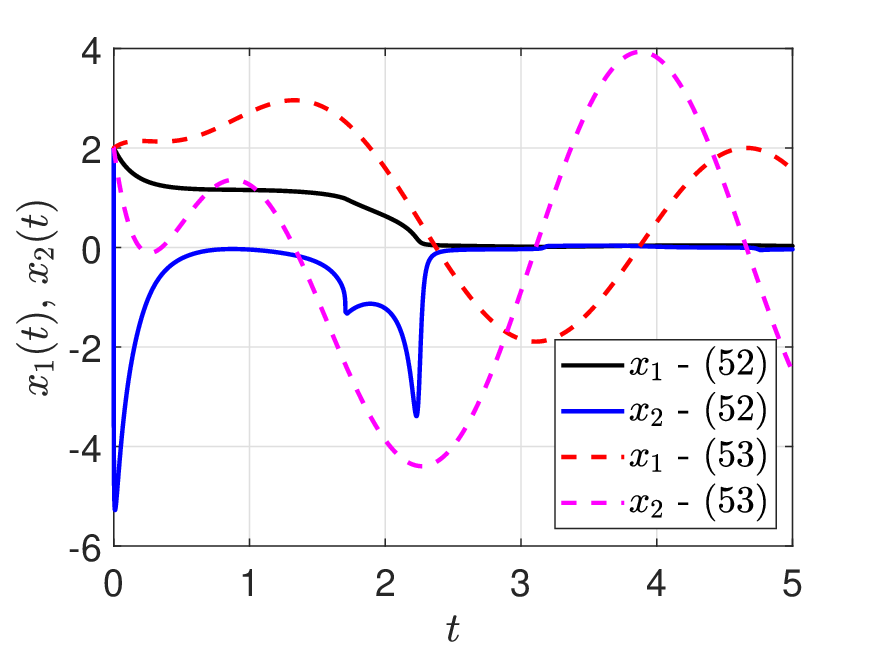}}
\caption{The graphs of solutions of \eqref{eq_lin_plant_ex1} using control laws \eqref{eq_lin_plant_CL_ex11} and \eqref{eq_lin_plant_CL_ex12} for $x(0)=col\{2,2\}$ and additive disturbance $20\sin(2t)$.}
\label{Fig_23_control}
\end{figure}


\section{Conclusions}
\label{Concl}

Three lemmas for fast convergence at a fixed-time of solutions of nonlinear dynamical systems are proposed for which special conditions are satisfied on the derivative of the quadratic function calculated along the system solutions.
Two lemmas are about the convergence of solutions to zero at a fixed-time, additional lemma is about convergence of solutions in a given region.
To achieve fast convergence, a negative power is applied to the derivative of the quadratic function within a specific domain of system evolution.
The application of the proposed results to the design of the control law for arbitrary order linear plants using the backstepping method is considered.
All the proposed results demonstrated fast convergence to zero compared to some existing control schemes.





\end{document}